\documentclass{LMCS}

\usepackage{amssymb,latexsym,amsmath,graphicx,hyperref}

\theoremstyle{plain}

\newtheorem{theorem}[thm]{Theorem}
\renewcommand{\(}{\begin{eqnarray*}}
\renewcommand{\)}{\end{eqnarray*}}
\newcommand\cl{\mathsf{cl}}

\newcommand\len[1]{|\kern1pt#1\kern1pt|}
\newcommand\llen[1]{\parallel\kern-1pt#1\kern-1pt\parallel}
\newcommand\subs{\mathrel{\subseteq}}
\newcommand\Imp{\mathrel{\Rightarrow}}
\newcommand\imp{\mathrel{\rightarrow}}
\newlength\reasonwidth
\newcommand\reasoning[1]{\def\longest{#1}\settowidth{\reasonwidth}{$\displaystyle\longest$}\addtolength{\reasonwidth}{5mm}}
\newcommand\reason[2]{\makebox[\reasonwidth][l]{$\displaystyle{#1}$}\mbox{#2}}
\renewcommand\star{^{\textstyle *}}
\newcommand\complexes{{\mathbb C}}
\newcommand\rationals{{\mathbb Q}}
\newcommand\naturals{{\mathbb N}}
\newcommand\eps{\varepsilon}
\newcommand\bx[1]{\mbox{\rm\texttt[}#1\mbox{\rm\texttt]}}
\newcommand\romanize{\def\theenumi{\roman{enumi}}\def\labelenumi{{\rm(\theenumi)}}}
\newcommand\set[2]{\{#1 \mid #2\}}
\newcommand\setcompl[1]{\mbox{$\sim$}\kern1pt{#1}}
\newcommand\meet\wedge
\renewcommand\phi\varphi
\newcounter{linenum}
\newcommand\linenum{\hspace{20pt}\makebox[0pt][l]{\tiny\textsc{\arabic{linenum}}}\hspace{16pt}\addtocounter{linenum}1}
\newlength\indlength
\newcommand\ind[1]{\hspace{#1\indlength}}
\font\tencmtt=cmtt10
\catcode`\<=1\catcode`\>=2\catcode`\{=12\catcode`\}=12\relax
\def\lbr<<\tencmtt {>>
\def\rbr<<\tencmtt }>>
\catcode`\<=12\catcode`\>=12\catcode`\{=1\catcode`\}=2\relax
\newenvironment{algorithm}{\begin{small}\begin{flushleft}\setcounter{linenum}1\spaceskip=4pt\baselineskip=13pt}{\end{flushleft}\end{small}}

\def\doi{3 (4:8) 2007}
\lmcsheading%
{\doi}
{1--14}
{}
{}
{May.~\phantom{0}2, 2007}
{Nov.~\phantom{0}9, 2007}
{}
\begin{document}

\title[Coinductive Proof Principles for Stochastic Processes]{Coinductive Proof Principles for Stochastic Processes\rsuper*}

\author[D.~Kozen]{Dexter Kozen}	
\address{Department of Computer Science\\
Cornell University\\
Ithaca, New York 14853-7501, USA}	
\email{kozen@cs.cornell.edu}  







\titlecomment{{\lsuper*}A preliminary version of this paper appeared as \cite{K06b}.}


\begin{abstract}
  \noindent We give an explicit coinduction principle for recursively-defined stochastic processes.  The principle applies to any closed property, not just equality, and works even when solutions are not unique.  The rule encapsulates low-level analytic arguments, allowing reasoning about such processes at a higher algebraic level.  We illustrate the use of the rule in deriving properties of a simple coin-flip process.
\end{abstract}

\maketitle

\section{Introduction}
\label{sec:intro}

Coinduction has been shown to be a useful tool in functional programming.  Streams, automata, concurrent and stochastic processes, and recursive types have been successfully analyzed using coinductive methods; see \cite{BarwiseMoss96,GupJagPan04,DesGupJagPan02,Rutten03,Gordon94} and references therein.

Most approaches emphasize the relationship between coinduction and bisimulation.  In Rutten's treatment \cite{Rutten03} (see also \cite{Gordon94,BarwiseMoss96}), the coinduction principle states that under certain conditions, two bisimilar processes must be equal.  For example, to prove the equality of infinite streams $\sigma=\texttt{merge}(\texttt{split}(\sigma))$, where \texttt{merge} and \texttt{split} satisfy the familiar coinductive definitions
\(
\mathtt{merge}(a::\sigma,\tau) &=& a::\mathtt{merge}(\tau,\sigma)\\
\mathtt{\#1}(\mathtt{split}(a::b::\rho)) &=& a::\mathtt{\#1}(\mathtt{split}(\rho))\\
\mathtt{\#2}(\mathtt{split}(a::b::\rho)) &=& b::\mathtt{\#2}(\mathtt{split}(\rho)),
\)
it suffices to show that the two streams are bisimilar.  An alternative view is that certain systems of recursive equations over a certain algebraic structure have unique solutions.  Desharnais et al.~\cite{DesGupJagPan02,GupJagPan04} study bisimulation in a probabilistic context.  They are primarily interested in the approximation of one process with another.  Again, they focus on bisimulation, but do not formulate an explicit coinduction rule.  

In this paper we introduce a generalization of the coinduction principle that applies to other properties besides equations and to situations in which solutions are not unique.  We illustrate its use with an extended example that demonstrates how the rule encapsulates low-level analytic arguments involving convergent sequences in its proof of soundness, thereby allowing reasoning about such processes at a higher algebraic level.

\section{An Example}
\label{sec:example}

Consider the following procedure for simulating a coin of arbitrary real bias $q$, $0\leq q\leq 1$, with a coin of arbitrary real bias $p$, $0 < p \leq 1/2$.  We assume unit-time exact arithmetic on real numbers.
\begin{algorithm}
\linenum\texttt{boolean} \textsc{Qflip}($q$) \lbr\\
\linenum\ind 1\texttt{if} ($q > p$) \lbr\\
\linenum\ind 2\texttt{if} (\textsc{Pflip}()) \texttt{return true};\\
\linenum\ind 2\texttt{else return} \textsc{Qflip}($(q-p)/(1-p)$);\\
\linenum\ind 1\rbr\ \texttt{else} \lbr\\
\linenum\ind 2\texttt{if} (\textsc{Pflip}()) \texttt{return} \textsc{Qflip}($q/p$);\\
\linenum\ind 2\texttt{else return false};\\
\linenum\ind 1\rbr\\
\linenum\rbr
\end{algorithm}
Intuitively, if $q > p$ and the bias-$p$ coin flip returns heads (\texttt{true}), which occurs with probability $p$, then we halt and output heads; this gives a fraction $p/q$ of the desired probability $q$ of heads of the simulated bias-$q$ coin.  If the bias-$p$ coin returns tails, which occurs with probability $1-p$, we rescale the problem appropriately and call \textsc{Qflip} tail-recursively.  Similarly, if $q\leq p$ and the bias-$p$ coin returns tails, then we halt and output tails; and if not, we rescale appropriately and call \textsc{Qflip} tail-recursively.

On any input $0\leq q\leq 1$, the probability of halting is 1, since the procedure halts with probability at least $p$ in each iteration.  The probability that \textsc{Qflip} halts and returns heads on input $q$ exists and satisfies the recurrence
\begin{eqnarray}
H(q) &=& \begin{cases}
p\cdot H(\frac qp), & \text{if } q\leq p,\\
p\ +\ (1-p)\cdot H(\frac{q-p}{1-p}), & \text{if } q > p.
\end{cases}
\label{eqn:prob0}
\end{eqnarray}
Now $H\star(q)=q$ is a solution to this recurrence, as can be seen by direct substitution.  There are uncountably many other solutions as well, but these are all unbounded (see Section \ref{sec:unbounded}).  Since $H\star$ is the unique bounded solution, it must give the probability of heads.

We can do the same for the expected running time.  Let us measure the expected number of calls to \textsc{Pflip} on input $q$.  The expectation exists and is uniformly bounded on the unit interval by $1/p$, the expected running time of a Bernoulli (coin-flip) process with success probability $p$.  From the program, we obtain the recurrence
\(
E_0(q)
&=& \begin{cases}
(1-p)\cdot 1\ +\ p\cdot(1+E_0(\frac qp)), & \text{if } q\leq p,\\
p\cdot 1\ +\ (1-p)\cdot (1 + E_0(\frac{q-p}{1-p})), & \text{if } q > p
\end{cases}\\[2pt]
&=& \begin{cases}
1\ +\ p\cdot E_0(\frac qp), & \text{if } q\leq p,\\
1\ +\ (1-p)\cdot E_0(\frac{q-p}{1-p}), & \text{if } q > p.
\end{cases}
\)
The unique bounded solution to this recurrence is
\begin{eqnarray}
E_0\star(q) &=& \frac qp\ +\ \frac{1-q}{1-p}.\label{eqn:exp1}
\end{eqnarray}
That it is a solution can be ascertained by direct substitution; uniqueness requires a further argument, which we will give later.  As before, there are uncountably many unbounded solutions, but since $E_0\star$ is the unique bounded solution, it must give the expected running time for any $q$.

The situation gets more interesting when we observe that slight modifications of the algorithm lead to noncontinuous fractal solutions with no simple characterizations like \eqref{eqn:exp1}.  The fractal behavior of stochastic processes has been previously observed in \cite{GupJagPan99}.

Currently, when $q > p$, we halt and output ``heads'' when \textsc{Pflip} gives heads, which occurs with probability $p$.  But note that we can save some time when $q\geq 1-p$.  In that case, we can halt and report heads if \textsc{Pflip} gives tails, which occurs with the larger probability $1-p$.  This allows us to take off a larger fraction of the remaining ``heads'' weight of the bias-$q$ coin.  If \textsc{Pflip} gives tails, we must still rescale, but the rescaling function is different.  The new code is in lines 2--4.
\begin{algorithm}
\linenum\texttt{boolean} \textsc{Qflip}($q$) \lbr\\
\linenum\ind 1\texttt{if} ($q \geq 1-p$) \lbr\\
\linenum\ind 2\texttt{if} (\textsc{Pflip}()) \texttt{return} \textsc{Qflip}($(q-(1-p))/p$);\\
\linenum\ind 2\texttt{else return true};\\
\linenum\ind 1\rbr\ \texttt{else if} ($q > p$) \lbr\\
\linenum\ind 2\texttt{if} (\textsc{Pflip}()) \texttt{return true};\\
\linenum\ind 2\texttt{else return} \textsc{Qflip}($(q-p)/(1-p)$);\\
\linenum\ind 1\rbr\ \texttt{else} \lbr\\
\linenum\ind 2\texttt{if} (\textsc{Pflip}()) \texttt{return} \textsc{Qflip}($q/p$);\\
\linenum\ind 2\texttt{else return false};\\
\linenum\ind 1\rbr\\
\linenum\rbr
\end{algorithm}
The recurrence for the expected running time is
\begin{eqnarray}
E_1(q) &=& 1 + r(q)E_1(f_1(q)),\label{eqn:fractal1}
\end{eqnarray}
where
\begin{align}
f_1(q)\ &=\ 
\begin{cases}
\frac qp, & \mbox{if $q\leq p$}\\[1ex]
\frac{q-p}{1-p}, & \mbox{if $p < q < 1-p$}\\[1ex]
\frac{q-(1-p)}p, & \mbox{if $q \geq 1-p$}
\end{cases}\label{eqn:f1def}\\[1ex]
r(q)\ &=\ 
\begin{cases}
1-p, & \mbox{if $p < q < 1-p$}\\
p, & \mbox{otherwise.}
\end{cases}\label{eqn:rdef}
\end{align}
Again, there is a unique bounded solution
\(
E_1\star(q) &=& \sum_{n=0}^\infty\prod_{j=0}^{n-1} r(f_1^j(q)),
\)
but there is no longer a nice algebraic characterization like \eqref{eqn:exp1}.  The solution for $p=1/4$ is the noncontinuous fractal shown in Fig.~\ref{fig:fractal1}, shown compared to the straight line $E_0\star$ running from 4/3 to 4.
\begin{figure}[ht]
\begin{center}
\includegraphics[width=240pt,height=150pt]{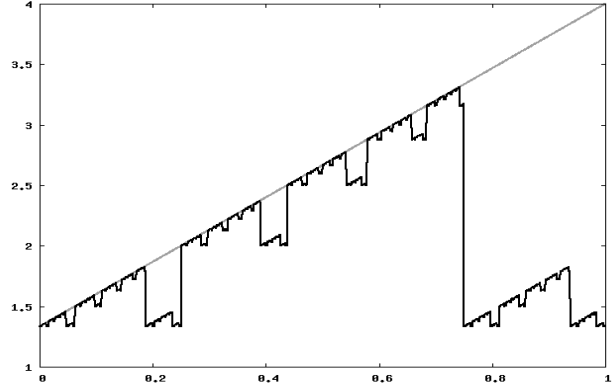}
\caption{Fractal solution of \eqref{eqn:fractal1}}
\label{fig:fractal1}
\end{center}
\end{figure}
The large discontinuity at $q = 1-p = 3/4$ is due to the modification of the algorithm for $q\geq 1-p$, and this discontinuity is propagated everywhere by the recurrence.

Fig.~\ref{fig:fractal1} and intuition dictate that $E_1\star\leq E_0\star$, but how do we prove this?  Not by induction, because there is no basis.  One might briefly imagine that it is because the second process halts no later than the first on any predetermined sequence of coin flips, but there are trivial counterexamples.  An analytic argument involving convergence of sequences seems inevitable.

However, there is a simpler alternative.  It will follow from our coinductive proof principle that to conclude $E_1\star\leq E_0\star$, it suffices to show that $\tau(E_1)(q)\leq E_0\star(q)$ whenever $E_1(f_1(q))\leq E_0\star(f_1(q))$, where $\tau$ is a suitably defined operator representing the unwinding of the recurrence \eqref{eqn:fractal1} once.  This property is easily checked algebraically, and no analysis is necessary.

We can modify the algorithm further to try to achieve more savings.  If $1/2 < q < 1-p$, it would seem to our advantage to remove $p$ from the tail probability of $q$ rather than from the head probability.  The intuition behind this heuristic is that that when $q$ is in one of the regions $\bx{0,p}$ or $\bx{1-p,1}$, we can halt in the next step with the higher probability $1-p$.  If $q > 1/2$, then the proposed new action will cause $q$ to move to the right toward the closer good region $\bx{1-p,1}$ instead of to the left, thereby getting to a good region faster.  The new code is in lines 5--7.
\begin{algorithm}
\linenum\texttt{boolean} \textsc{Qflip}($q$) \lbr\\
\linenum\ind 1\texttt{if} ($q \geq 1-p$) \lbr\\
\linenum\ind 2\texttt{if} (\textsc{Pflip}()) \texttt{return} \textsc{Qflip}($(q-(1-p))/p$);\\
\linenum\ind 2\texttt{else return true};\\
\linenum\ind 1\rbr\ \texttt{else if} ($q > 1/2$) \lbr\\
\linenum\ind 2\texttt{if} (\textsc{Pflip}()) \texttt{return false};\\
\linenum\ind 2\texttt{else return} \textsc{Qflip}($q/(1-p)$);\\
\linenum\ind 1\rbr\ \texttt{else if} ($q > p$) \lbr\\
\linenum\ind 2\texttt{if} (\textsc{Pflip}()) \texttt{return true};\\
\linenum\ind 2\texttt{else return} \textsc{Qflip}($(q-p)/(1-p)$);\\
\linenum\ind 1\rbr\ \texttt{else} \lbr\\
\linenum\ind 2\texttt{if} (\textsc{Pflip}()) \texttt{return} \textsc{Qflip}($q/p$);\\
\linenum\ind 2\texttt{else return false};\\
\linenum\ind 1\rbr\\
\linenum\rbr
\end{algorithm}
The recurrence is
\begin{eqnarray}
E_2(q) = 1 + r(q)E_2(f_2(q))\label{eqn:fractal2}
\end{eqnarray}
with
\(
f_2(q) &=&
\begin{cases}
\frac qp, & \text{if } q\leq p,\\[1ex]
\frac{q-p}{1-p}, & \text{if } p < q\leq 1/2,\\[1ex]
\frac q{1-p}, & \text{if } 1/2 < q < 1-p,\\[1ex]
\frac{q-(1-p)}p, & \text{if } q \geq 1-p,
\end{cases}
\)
and $r(q)$ as given in \eqref{eqn:rdef}.  The symmetric fractal solution
\(
E_2\star(q) &=& \sum_{n=0}^\infty\prod_{j=0}^{n-1} r(f_2^j(q))
\)
is shown in Fig.~\ref{fig:fractal2}.
\begin{figure}[ht]
\begin{center}
\includegraphics[width=240pt,height=150pt]{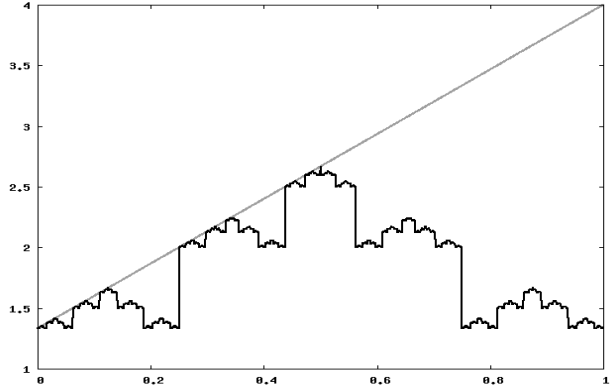}
\caption{Fractal solution of \eqref{eqn:fractal2}}
\label{fig:fractal2}
\end{center}
\end{figure}

Intuition seems to say that this solution should be at least as good as $E_1\star$, but it turns out that this is not always the case.  By unwinding the recurrences a few steps and using the lower bound
\(
E_2\star(q) &\geq& \sum_{n=0}^\infty p^n\ \ =\ \ \frac 1{1-p},
\)
it can be shown that for $p=1/4$,
\(
E_1\star(11/20) &=& 5/2\ \ =\ \ 2.5\\
E_2\star(11/20) &\geq& 323/128\ \ \approx\ \ 2.5234375\ldots~.
\)
Moreover, this inversion holds on an open interval containing 11/20 and countably many other open intervals.

One might ask whether there is a slight modification of $E_2\star$ that is everywhere better than $E_1\star$.  The answer is yes: take the breakpoint not at 1/2, but at
\(
c &=& \max((1-p)^2,1-(1-p)^2),
\)
provided $p\leq (1-p)^2$.  For $p=1/4$, this gives $c=9/16$.  Now the recurrence is
\begin{eqnarray}
E_3(q) &=& 1 + r(q)E_3(f_3(q))\label{eqn:fractal3}
\end{eqnarray}
with
\(
f_3(q) &=&
\begin{cases}
\frac qp, & \text{if } q\leq p,\\[1ex]
\frac{q-p}{1-p}, & \text{if } p < q\leq c,\\[1ex]
\frac q{1-p}, & \text{if } c < q < 1-p,\\[1ex]
\frac{q-(1-p)}p, & \text{if } q\geq 1-p.
\end{cases}
\)
(Since $p\leq (1 - p)^2$ implies $p\leq (3 - \sqrt 5)/2 \approx 0.382$, this modification will not work for all $p$.)

Now we wish to show that $E_3\star\leq E_1\star$ on the whole unit interval.  Note that we are comparing two nowhere-differentiable functions\footnote{Hermite and Poincar\'e eschewed such functions, calling them a ``dreadful plague''.  Poincar\'e wrote: ``Yesterday, if a new function was invented, it was to serve some purpose; today, they are invented only to debunk the arguments of our predecessors, and they will never have any other use.''}; we have no nice algebraic description of them save as solutions of the recurrences $E_i(q) = 1 + r(q)E_i(f_i(q))$.  However, we can prove the desired inequality purely algebraically using the coinductive principle below, without recourse to analysis.  We outline a proof below, after we have stated the principle and proved its validity.

\section{A Coinduction Principle}
\label{sec:coind}

In this section we state and prove a coinduction principle that will allow us to derive properties of stochastic processes.  The version we will use is most conveniently formulated in terms of bounded linear operators on a Banach space (complete normed linear space), but is closely related to a coinduction principle that holds in arbitrary complete metric spaces.  We treat the metric version first.

Let $(V,d)$ be a complete metric space.  A function $\tau:V\imp V$ is \emph{contractive} if there exists a $c < 1$ such that for all $u,v\in V$, $d(\tau(u),\tau(v)) \leq c\cdot d(u,v)$.  The value $c$ is called the \emph{constant of contraction}.  A continuous function $\tau$ is said to be \emph{eventually contractive} if $\tau^n$ is contractive for some $n\geq 1$.  Contractive maps are uniformly continuous, and by the Banach fixpoint theorem, any such map has a unique fixpoint in $V$.

The fixpoint of a contractive map $\tau$ can be constructed explicitly as the limit of a Cauchy sequence $u,\tau(u),\tau^2(u),\ldots$ starting at any point $u\in V$.  The sequence is Cauchy; one can show by elementary arguments that
\(
d(\tau^{n+m}(u),\tau^n(u)) &\leq& c^n(1-c^m)(1-c)^{-1}\cdot d(\tau(u),u).
\)
Since $V$ is complete, the sequence has a limit $u\star$, which by continuity must be a fixpoint of $\tau$.  Moreover, $u\star$ is unique: if $\tau(u)=u$ and $\tau(v)=v$, then
\(
d(u,v) = d(\tau(u),\tau(v)) \leq c\cdot d(u,v) &\Imp& d(u,v)=0,
\)
therefore $u=v$.

Eventually contractive maps also have unique fixpoints.  If $\tau^n$ is contractive, let $u\star$ be the unique fixpoint of $\tau^n$.  Then $\tau(u\star)$ is also a fixpoint of $\tau^n$.  But then $d(u\star,\tau(u\star)) = d(\tau^n(u\star),\tau^{n+1}(u\star)) \leq c\cdot d(u\star,\tau(u\star))$, therefore $u\star$ is also a fixpoint of $\tau$.

In this framework, the coinduction rule takes the following simple form.  If $\phi$ is a closed nonempty subset of a complete metric space $V$, and if $\tau$ is an eventually contractive map on $V$ that preserves $\phi$, then the unique fixpoint $u\star$ of $\tau$ is in $\phi$.  Expressed as a proof rule, this says for $\phi$ a closed property,
\begin{equation}
\frac{\exists u\ \phi(u) \qquad \forall u\ \phi(u) \Imp \phi(\tau(u))}{\phi(u\star)}\label{eqn:coind1}
\end{equation}
This is quite easily proved.  Since $\phi$ is nonempty, it contains a point $u$.  Since $\phi$ is preserved by $\tau$, all elements of the sequence $\tau^n(u)$ are contained in $\phi$.  Finally, since $\phi$ is closed, the fixpoint $u\star$ is contained in $\phi$, since it is the limit of a Cauchy subsequence.

For our purposes, the coinduction principle is most conveniently expressed in the following form.  This form makes clear how the principle allows analytic arguments to be replaced by simpler algebraic ones.  See \cite{DunfordSchwartz57} for the necessary background.

Let $B$ be a Banach space (complete normed linear space) over $\complexes$ and let $R$ be a bounded linear operator on $B$ (\emph{bounded} is synonymous with \emph{continuous} for linear operators on $B$).  The \emph{spectrum} of $R$, denoted $\sigma(R)$, is the set of complex numbers $\lambda$ such that $\lambda I - R$ is not invertible.  The \emph{spectral radius} of $R$ is
\begin{eqnarray}
\sup_{\lambda\in\sigma(R)}\len{\lambda} &=& \inf_n \sqrt[n]{\llen{R^n}},\label{eqn:specradius}
\end{eqnarray}
where
\(
\llen R &=& \sup_{\parallel x\parallel=1} \llen{R(x)}.
\)
Suppose that $I-R$ is invertible; that is, $1\not\in\sigma(R)$.  Let $a\in B$.  Then there is a unique solution $e\star$ of the equation $e = a + Re$, namely $e\star = (I-R)^{-1}a$.

\begin{theorem}
\label{thm:converge1}
Consider the affine operator $\tau(e) = a + Re$, where $R$ is a bounded linear operator with spectral radius strictly less than $1$.  Let $\phi\subs B$ be a closed nonempty region preserved by $\tau$.  Then $e\star\in\phi$.
\end{theorem}
\begin{proof}
By \eqref{eqn:specradius}, if the spectral radius of $R$ is less than $1$, then $R$ is eventually contractive; that is, there exists $n$ such that $\llen{R^n} < 1$.  Then $\tau$ is also eventually contractive, since
\(
\llen{\tau^n(e) - \tau^n(e')} &=& \llen{\sum_{i=0}^{n-1}R^i(a) + R^n(e) - \sum_{i=0}^{n-1}R^i(a) - R^n(e')}\\
&=& \llen{R^n(e - e')}\\
&\leq& \llen{R^n}\cdot\llen{e - e'}.
\)
It follows from \eqref{eqn:coind1} that the unique fixpoint of $\tau^n$ is contained in $\phi$.  But this fixpoint must be $e\star$, since $e\star$ is a fixpoint of $\tau$. 
\end{proof}

Restated as a proof rule, Theorem \ref{thm:converge1} takes the following form:
\begin{theorem}
\label{thm:converge2}
Let $\tau$ be as in Theorem \ref{thm:converge1}.  Let $\phi$ be a closed property.  The following rule is valid:
\begin{eqnarray}
\frac{\exists e\ \phi(e)\qquad\forall e\ \phi(e) \Imp \phi(\tau(e))}{\phi(e\star)}.\label{eqn:rule1}
\end{eqnarray}
More generally, for any $n\geq 1$,
\begin{eqnarray}
\frac{\exists e\ \phi(e)\qquad\forall e\ \phi(e) \Imp \phi(\tau^n(e))}{\phi(e\star)}.\label{eqn:rule2}
\end{eqnarray}
\end{theorem}
\begin{proof}
The rule \eqref{eqn:rule1} is just a restatement of Theorem \ref{thm:converge1}.  The rule \eqref{eqn:rule2} follows by applying \eqref{eqn:rule1} to the closed property $\psi(e) = \bigvee_{i=0}^{n-1} \phi(\tau^i(e))$.  This is a closed property because $\tau$ is continuous on $B$.
\end{proof}

For example, to show that $E_1\star\leq E_0\star$ using the rule \eqref{eqn:rule1}, we take $B$ to be the space of bounded real-valued functions on the unit interval, $a=\lambda x.1$, $R:B\imp B$ the bounded linear operator 
\begin{eqnarray}
R &=& \lambda E.\lambda q.r(q)E(f_1(q))\label{eqn:Rdef}
\end{eqnarray}
with spectral radius $1-p$, $\phi(E)$ the closed property
\(
\forall q\ \ E(q) &\leq& \frac qp + \frac{1-q}{1-p},
\)
and
\begin{eqnarray}
\tau(E) &=& \lambda q.(1+r(q)E(f_1(q)))\ \ =\ \ \lambda q.(1+RE(q)),\label{eqn:taudef}
\end{eqnarray}
where $f_1$ and $r$ are as given in \eqref{eqn:f1def} and \eqref{eqn:rdef}.  That the spectral radius of $R$ is at most $1-p$ follows immediately from \eqref{eqn:specradius}, since
\(
\llen R &=& \sup_{\parallel E\parallel=1} \sup_q \len{RE(q)}\ \ =\ \ \sup_{\parallel E\parallel=1} \sup_q \len{r(q)E(f_1(q))}\ \ \leq\ \ 1-p.
\)
That it is exactly $1-p$ requires a further argument, which we defer to Section \ref{sec:unbounded}.

Now the desired conclusion is
\begin{eqnarray}
\forall q\ \ E_1\star(q) &\leq& \frac qp + \frac{1-q}{1-p},\label{eqn:exm1}
\end{eqnarray}
and the two premises we must establish are
\begin{eqnarray}
\exists E\ \ \forall q\ E(q) &\leq& \frac qp + \frac{1-q}{1-p},\label{eqn:premise1}
\end{eqnarray}
\begin{eqnarray}
\forall E\ \ \left(\forall q\ E(q) \leq \frac qp + \frac{1-q}{1-p}\right.
&\Imp& \left.\forall q\ \tau(E)(q) \leq \frac qp + \frac{1-q}{1-p}\right).\label{eqn:premise2}
\end{eqnarray}
The premise \eqref{eqn:premise1} is trivial; for example, take $E = \lambda q.0$.  For \eqref{eqn:premise2}, let $E$ be arbitrary.  We wish to show that
\begin{eqnarray}
\forall q\ E(q) \leq \frac qp + \frac{1-q}{1-p}
&\Imp& \forall q\ \tau(E)(q) \leq \frac qp + \frac{1-q}{1-p}.\label{eqn:a2}
\end{eqnarray}
Picking $q$ arbitrarily on the right-hand side and then specializing the left-hand side at $f_1(q)$, it suffices to show
\begin{eqnarray}
E(f_1(q)) \leq \frac{f_1(q)}p + \frac{1-f_1(q)}{1-p}
&\Imp& \tau(E)(q) \leq \frac qp + \frac{1-q}{1-p}.\label{eqn:a5}
\end{eqnarray}
Substituting the definition of $\tau$, we need to show
\begin{eqnarray}
E(f_1(q)) \leq \frac{f_1(q)}p + \frac{1-f_1(q)}{1-p}
&\Imp& 1+r(q)E(f_1(q)) \leq \frac qp + \frac{1-q}{1-p}.\label{eqn:a6}
\end{eqnarray}
The proof breaks into three cases, depending on whether $q\leq p$, $p < q < 1-p$, or $q \geq 1-p$.  In the first case, $f_1(q)= q/p$ and $r(q)=p$.  Then \eqref{eqn:a6} becomes
\(
E(\frac qp) \leq \frac q{p^2} + \frac{1-\frac qp}{1-p} &\Imp& 1+pE(\frac qp) \leq \frac qp + \frac{1-q}{1-p}.
\)
But
\(
1+pE(\frac qp) &\leq& 1 + p(\frac q{p^2} + \frac{1-\frac qp}{1-p})\ \ =\ \ \frac qp + \frac{1-q}{1-p}.
\)
The remaining two cases are equally straightforward.  The last case, $q \geq 1-p$, uses the fact that $p\leq 1/2$.

\medskip

One can also prove closed properties of more than one function $E$.  For example, as promised, we can show that $E_3\star\leq E_1\star$ whenever
\(
\max((1-p)^2,1-(1-p)^2) &\leq& c\ \ \leq\ \ 1-p.
\)
For this application, $B$ is the space of pairs $(E,E')$, where $E$ and $E'$ are bounded real-valued functions on the unit interval, $a=(\lambda x.1,\lambda x.1)$, and $R:B\imp B$ is the bounded linear operator
\(
R(E,E') &=& (\lambda q.r(q)E(f_3(q)),\lambda q.r(q)E'(f_1(q)))
\)
with spectral radius $1-p$.  The closed property of interest is $E\leq E'$, but we need the stronger coinduction hypothesis
\begin{eqnarray}
\phi(E,E') &=& \parbox[l]{14pt}{$\forall q$} E(q)\leq E'(q)\label{eqn:IH1}\\
&& \hspace{14pt}\wedge\ \ E(q)\geq \frac 1{1-p}\label{eqn:IH2}\\
&& \hspace{14pt}\wedge\ \ p < q < 1-p\ \ \Imp\ \ E'(q)\geq 2\label{eqn:IH3}\\
&& \hspace{14pt}\wedge\ \ E'(q)\leq \frac qp + \frac{1-q}{1-p}\label{eqn:IH4}\\
&& \hspace{14pt}\wedge\ \ 0 \leq q \leq p\ \ \Imp\ \ E(q) = E(q+1-p).\label{eqn:IH5}
\end{eqnarray}
Equivalent to \eqref{eqn:IH5} is the statement
\begin{eqnarray}
1-p \leq q \leq 1 &\Imp& E(q) = E(q-(1-p)).\label{eqn:IH6}
\end{eqnarray}
There certainly exist $(E,E')$ satisfying $\phi$.  We have also already argued that coinduction hypothesis \eqref{eqn:IH4} is preserved by $\tau$.  The argument for \eqref{eqn:IH2} is similar.  For \eqref{eqn:IH5}, if $0 \leq q\leq p$, then
\(
1-p &\leq& q+1-p\ \ \leq\ \ 1,
\)
therefore
\(
r(q) &=& r(q+1-p)\ \ =\ \ p\\
f_3(q) &=& \frac qp\\
f_3(q+1-p) &=& \frac{(q+1-p)-(1-p)}{p}\ \ =\ \ \frac qp.
\)
It follows that
\(
1+r(q)E(f_3(q)) &=& 1+r(q+1-p)E(f_3(q+1-p))\ \ =\ \ 1+pE(q/p).
\)

For \eqref{eqn:IH3}, if $p < q < 1-p$, then
\(
r(q) &=& 1-p\\
E'(f_1(q)) &\geq& \frac 1{1-p}
\)
by the coinduction hypotheses \eqref{eqn:IH1} and \eqref{eqn:IH2}, thus
\(
1+r(q)E'(f_1(q)) &\geq& 1+(1-p)\frac 1{1-p}\ \ =\ \ 2.
\)

Finally, for \eqref{eqn:IH1}, we wish to show
\(
1+r(q)E(f_3(q)) &\leq& 1+r(q)E'(f_1(q)),
\)
or equivalently,
\begin{eqnarray}
E(f_3(q)) &\leq& E'(f_1(q)).\label{eqn:last}
\end{eqnarray}
Since $f_1$ and $f_3$ coincide except in the range $c < q < 1-p$, we need only show \eqref{eqn:last} for $q$ in this range.

It follows from the assumptions in effect that
\(
p &<& f_1(q)\ \ =\ \ \frac{q-p}{1-p}\ \ <\ \ 1-p\ \ <\ \ \frac q{1-p}\ \ =\ \ f_3(q),
\)
thus
\reasoning{\frac{\frac q{1-p} - (1-p)}{p} + \frac{1-(\frac q{1-p} - (1-p))}{1-p}}
\(
E(f_3(q))
&=& \reason{E(\frac q{1-p} - (1-p))}{by \eqref{eqn:IH5}, in the form \eqref{eqn:IH6}}\\
&\leq& \reason{\frac{\frac q{1-p} - (1-p)}{p} + \frac{1-(\frac q{1-p} - (1-p))}{1-p}}{by \eqref{eqn:IH4}}\\
&=& (\frac q{1-p} - 1)\frac{1-2p}{p(1-p)} + 2\\
&\leq& \reason{2}{since $p,q\leq 1-p$}\\
&\leq& \reason{E'(f_1(q))}{by \eqref{eqn:IH3}.}
\)

We can conclude from the coinduction rule that $\phi(E_3\star,E_1\star)$.  Note that nowhere in this proof did we use any analytic arguments.  All the necessary analysis is encapsulated in the proof of Theorem \ref{thm:converge1}.

\medskip

As a final application, we show how to use the coinductive proof rule \eqref{eqn:rule1} of Theorem \ref{thm:converge2} to argue that for $p < 1/2$, the function $E_1\star$ is nowhere differentiable.  We do this by showing that $E_1\star$ has a dense set of discontinuities on the unit interval.

First we show that $E_1\star$ has discontinuities at $p$ and $1-p$.  We know from clause \eqref{eqn:IH3} of the previous argument that for all $q$ in the range $p < q < 1-p$,
\begin{eqnarray}
E_1\star(q) &\geq& 2.\label{eqn:discont1}
\end{eqnarray}
Also, by \eqref{eqn:exm1}, we have that $E_1\star(q) \leq 1/p$ for all $q$.  Then for $\eps < p^2$, unwinding the defining recurrence \eqref{eqn:fractal1} for $E_1\star$ twice yields
\begin{eqnarray}
E_1\star(1-p+\eps) &=& 1+p+p^2E_1\star(\frac{\eps}{p^2})\ \ \leq\ \ 1+2p\label{eqn:discont2}\\
E_1\star(p-\eps) &=& 1+p+p^2E_1\star(1-\frac{\eps}{p^2})\ \ \leq\ \ 1+2p.\label{eqn:discont3}
\end{eqnarray}
Since $1 + 2p < 2$, \eqref{eqn:discont1}--\eqref{eqn:discont3} imply that $E_1\star$ has discontinuities at $p$ and $1-p$.

Finally, we show that every nonempty open interval contains a discontinuity.  Suppose for a contradiction that $E_1\star$ is continuous on a nonempty open interval $(a,b)$.  The interval $(a,b)$ can contain neither $p$ nor $1-p$, so the entire interval must be contained in one of the three regions $(0,p)$, $(p,1-p)$, or $(1-p,1)$.

Suppose it is contained in $(0,p)$.  Then
\(
E_1\star(q) &=& 1 + pE_1\star(q/p)
\)
for $a < q < b$, thus
\(
E_1\star(q/p) &=& (E_1\star(q) - 1)/p
\)
for $a/p < q/p < b/p$, so $E_1\star$ is also continuous on the interval $(a/p,b/p)$.  But the length of this interval is $(b-a)/p$, thus we have produced a longer interval on which $E_1\star$ is continuous.

A similar argument holds if $(a,b)$ is contained in one of the intervals $(p,1-p)$ or $(1-p,1)$.  In each of these three cases, we can produce an interval of continuity that is longer than $(a,b)$ by a factor of at least $1/(1-p)$.  This process can be repeated at most $\log (b-a)/\log (1-p)$ steps before the interval must contain one of the discontinuities $p$ or $1-p$.  This is a contradiction.

\section{Unbounded Solutions}
\label{sec:unbounded}

That these coinductive proofs have no basis is reflected in the fact that there exist unbounded solutions in addition to the unique bounded solutions.  All unbounded solutions are necessarily noncontinuous, because any continuous solution on a closed interval is bounded.

Theorem \ref{thm:converge1} does not mention these unbounded solutions, because they live outside the Banach space $B$.  Nevertheless, it is possible to construct unbounded solutions to any of the above recurrences.  All these recurrences are of the form
\begin{eqnarray}
E(q) &=& a + r(q)E(f(q)).\label{eqn:gen}
\end{eqnarray}
Let $G$ be the graph with vertices $q\in\bx{0,1}$ and edges $(q,f(q))$.  Note that every vertex in $G$ has outdegree 1.  Let $C$ be an undirected connected component of $G$.  One can show easily that the following are equivalent:
\begin{enumerate}
\romanize
\item
$C$ contains an undirected cycle;
\item
$C$ contains a directed cycle;
\item
for some $q\in C$ and $k>0$, $f^k(q)=q$.
\end{enumerate}
Call $C$ \emph{rational} if these conditions hold of $C$, \emph{irrational} otherwise.  For example, for $f_1$ given in \eqref{eqn:f1def}, the connected components of $0$ and $1$ are rational, since $f_1(0)=0$ and $f_1(1)=1$.  There are other rational components besides these; for example, if $p=1/4$, the component of $q=11/20$ is rational, since $f_1^2(11/20)=f_1^4(11/20)=1/5$.

Now any solution $E$ of \eqref{eqn:gen} must agree with the unique bounded solution $E\star$ on the rational components: if $f^k(q)=q$, then unwinding the recurrence $k$ times gives
\(
E(q) &=& a\sum_{n=0}^{k-1}\prod_{i=0}^{n-1}r(f^i(q)) + \left(\prod_{i=0}^{k-1}r(f^i(q))\right)E(q),
\)
therefore
\(
E(q) &=& \frac{a\sum_{n=0}^{k-1}\prod_{i=0}^{n-1}r(f^i(q))}{1-\prod_{i=0}^{k-1}r(f^i(q))}.
\)
But the values of $E$ on an entire connected component are uniquely determined by its value on a single element of the component, since $E(q)$ uniquely determines $E(f(q))$ and vice-versa.  Thus $E$ and $E\star$ must agree on the entire component.

We note in passing that this allows us to construct an $E$ such that $RE = (1-p)E$, where $R$ is the linear operator of \eqref{eqn:Rdef}, thereby establishing that the spectral radius of $R$ is $1-p$.  Take $E(1)=1$, then inductively define $E(q) = r(q)E(f_1(q))/(1-p)$ for all other $q$ in the component of 1 and $E(q) = 0$ otherwise.  Then $\llen E=1$, and
\(
RE(q) &=& r(q)E(f_1(q))\ \ =\ \ r(q)E(q)\cdot\frac{1-p}{r(q)}\ \ =\ \ (1-p)E(q).
\)

For an irrational component, since there are no cycles, it is connected as a tree.  We can freely assign an arbitrary value to an arbitrarily chosen element $q$ of the component, then extend the function to the entire component uniquely and without conflict.

For $f\in\{f_1,f_2,f_3\}$ of the examples of Section \ref{sec:example}, there always exists an irrational component.  This follows from the fact that if $f^k(q)=q$, then $q$ is a rational function of $p$; that is, $q$ is an element of the field $\rationals(p)$.  To see this, note that any $f^k(q)$ is of the form
\(
\frac q{p^m(1-p)^{k-m}}\ -\ r
\)
for some $0\leq m\leq k$ and $r\in\rationals(p)$.  This can be shown by induction on $k$.  Solving $f^k(q)=q$ for $q$ gives
\(
q &=& \frac{rp^m(1-p)^{k-m}}{1 - p^m(1-p)^{k-m}}\ \ \in\ \ \rationals(p).
\)
Thus the component of any real $q\not\in\rationals(p)$ is an irrational component.  There exist uncountably many such $q$, since $\rationals(p)$ is countable.  In fact, there are uncountably many irrational components, since each component is countable, and a countable union of countable sets is countable.  Moreover, it can be shown that if $q_1$ and $q_2$ are in the same component, then $\rationals(p,q_1)=\rationals(p,q_2)$.  This is because if $q_1$ and $q_2$ are in the same component, then $f^{k_1}(q_1)=f^{k_2}(q_2)$ for some $k_1,k_2\in\naturals$, so
\(
\frac{q_1}{p^{m_1}(1-p)^{k_1-m_1}}\ -\ r_1 &=& \frac{q_2}{p^{m_2}(1-p)^{k_2-m_2}}\ -\ r_2,
\)
therefore $q_1\in\rationals(p,q_2)$ and $q_2\in\rationals(p,q_1)$.

We have thus characterized all possible solutions.

\section{Why Is This Coinduction?}
\label{sec:why}

The reader may be curious why we have called the rule \eqref{eqn:rule1} a coinduction rule, since it may seem different from the usual forms of coinduction found in the literature.  The form of the rule and its use in applications certainly bears a resemblance to other versions in the literature, but to justify the terminology on formal grounds, we must exhibit a category of coalgebras and show that the rule \eqref{eqn:rule1} is equivalent to the assertion that a certain coalgebra is final in the category.  

Say we have a contractive map $\tau$ on a metric space $B$ and a nonempty closed subset $\phi\subs B$ preserved by $\tau$.  Define $\tau(\phi)=\set{\tau(s)}{s\in\phi}$.  Consider the category $C$ whose objects are the nonempty closed subsets of $B$ and whose arrows are the reverse set inclusions; thus there is a unique arrow $\phi_1\imp\phi_2$ iff $\phi_1\supseteq\phi_2$.  The map $\bar\tau$ defined by $\bar\tau(\phi)=\cl(\tau(\phi))$, where $\cl$ denotes closure in the metric topology, is an endofunctor on $C$, since $\bar\tau(\phi)$ is a nonempty closed set, and $\phi_1\supseteq\phi_2$ implies $\bar\tau(\phi_1)\supseteq \bar\tau(\phi_2)$.  A $\bar\tau$-coalgebra is then a nonempty closed set $\phi$ such that $\phi\supseteq \bar\tau(\phi)$; equivalently, such that $\phi\supseteq \tau(\phi)$.  The final coalgebra is $\{e\star\}$, where $e\star$ is the unique fixpoint of $\tau$.  The coinduction rule \eqref{eqn:rule1} says that $\phi\supseteq \tau(\phi)\Imp \phi\supseteq\{e\star\}$, which is equivalent to the statement that $\{e\star\}$ is final in the category of $\bar\tau$-coalgebras.

\section{Future Work}
\label{sec:future}

There is great potential in the use of proof principles similar to those of Theorem \ref{thm:converge2} for simplifying arguments involving probabilistic programs, stochastic processes, and dynamical systems.  Such rules encapsulate low-level analytic arguments, thereby allowing reasoning about such processes at a higher algebraic or logical level.  A few such applications have been described in the theory of streams, Markov chains and Markov decision processes, and non-well-founded sets \cite{KR07a}.  Other possible application areas are complex and functional analysis, the theory of linear operators, measure theory and integration, random walks, fractal analysis, functional programming, and probabilistic logic and semantics.

In particular, probabilistic programs can be modeled as measurable kernels $R(x,A)$, which can be interpreted as forward-moving measure transformers or backward-moving measurable function transformers \cite{DesGupJagPan99,K81c}. The expectation functions considered in this paper were uniformly bounded, but there are examples of probabilistic programs for which this is not true. It would be nice to find rules to handle these cases.

An intriguing open problem is whether the optimal strategy for the coin-flip process of Section \ref{sec:example} is decidable.  Specifically, given rational $p,q$, $0 < p,q < 1$, and a flip of the bias-$p$ coin, can we decide what action to take to minimize the expected running time?  It is known that $E_3\star$ is not necessarily optimal.

\section*{Acknowledgements}

Thanks to Terese Damh\o j Andersen, Lars Backstrom, Juris Hartmanis, Geoff Kozen, Prakash Panangaden, Jan Rutten, and the anonymous referees.  This work was supported in part by ONR Grant N00014-01-1-0968 and by NSF grant CCF-0635028.  The views and conclusions herein are those of the author and do not necessarily represent the official policies or endorsements of these organizations or the US government.

\bibliographystyle{plain}

\end{document}